\pdfoutput=1
\def\spaa{0} 
\ifnum\spaa=0
\documentclass[11pt]{article}
\usepackage{times}
\usepackage{amsthm}
\usepackage{fullpage}
\else
\documentclass[9pt]{sig-alternate}
\fi

\usepackage{xspace}

\usepackage{enumerate}
\usepackage{comment}
\usepackage{url}
\usepackage{subfig}
\usepackage{amsmath}
\usepackage{amssymb}
\usepackage{wrapfig}
\usepackage{graphicx}
\graphicspath{{./Figs/}}
\usepackage{paralist}
\usepackage[active]{srcltx}
\usepackage{algorithm}
\usepackage{algpseudocode}
\newtheorem{theorem}{Theorem}
\newtheorem{rem}[theorem]{Remark}

\newtheorem{lemma}[theorem]{Lemma}
\newtheorem{claim}[theorem]{Claim}
\newtheorem{proposition}[theorem]{Proposition}
\newtheorem{definition}[theorem]{Definition}

\usepackage[colorlinks,linkcolor=blue,filecolor=blue,citecolor=blue,urlcolor=blue,pdfstartview=FitH,plainpages=false]{hyperref}

\newcommand{\namedref}[2]{\hyperref[#2]{#1~\ref*{#2}}}
\newcommand{\sectionref}[1]{\namedref{Section}{#1}}

\newcommand{\theoremref}[1]{\namedref{Theorem}{#1}}
\newcommand{\defref}[1]{\namedref{Definition}{#1}}
\newcommand{\figureref}[1]{\namedref{Figure}{#1}}
\newcommand{\figref}[1]{\namedref{Figure}{#1}}
\newcommand{\claimref}[1]{\namedref{Claim}{#1}}
\newcommand{\lemmaref}[1]{\namedref{Lemma}{#1}}
\newcommand{\tableref}[1]{\namedref{Table}{#1}}

\newcommand{\propref}[1]{\namedref{Proposition}{#1}}
\newcommand{\Algref}[1]{\namedref{Algorithm}{#1}}
\newcommand{\lineref}[1]{\namedref{Line}{#1}}

\newcommand{\Lineref}[1]{\namedref{Line}{#1}}
\newcommand{\equalityref}[1]{\hyperref[#1]{Equality~\eqref{#1}}}
\newcommand{\inequalityref}[1]{\hyperref[#1]{Inequality~\eqref{#1}}}

\usepackage{color}
\usepackage{soul}

\newcommand{\hlc}[2][yellow]{ {\sethlcolor{#1} \hl{#2}} }
\newcommand{\moti}[1]{\hlc[green]{\textbf{MM}: #1}}

\algnewcommand\algorithmicswitch{\textbf{switch}}
\algnewcommand\algorithmiccase{\textbf{case}}
\algnewcommand\algorithmicassert{\texttt{assert}}
\algnewcommand\Assert[1]{\State \algorithmicassert(#1)}%

\algdef{SE}[SWITCH]{Switch}{EndSwitch}[1]{\algorithmicswitch\ #1\ \algorithmicdo}{\algorithmicend\ \algorithmicswitch}%
\algdef{SE}[CASE]{Case}{EndCase}[1]{\algorithmiccase\ #1}{\algorithmicend\ \algorithmiccase}%
\algtext*{EndSwitch}%
\algtext*{EndCase}%

\newtheorem{defn}[theorem]{Definition}
\newcommand{\INDSTATE}[1][1]{\State\quad}

\newcommand{\route}{\text{\sc{ipp}}\xspace}
\newcommand{\IPP}{\route}
\newcommand{\routenear}{\text{\sc{route-near}}\xspace}

\newcommand{\initroute}{\text{\sc{initial-route}}\xspace}
\newcommand{\detailedroute}{\text{\sc{detailed-route}}}

\newcommand{\opt}{\text{\textsc{opt}}}

\newcommand{\pmax}{p_{\max}}

\newcommand{\alg}{\text{\sc{alg}}}
\newenvironment{proof sketch}[1]{\noindent {\emph{Proof sketch of #1:}}}{\hfill \qed}

\newcommand{\Set}[1]{\left\{#1\right\}}

\def\DEF{\triangleq}
\newcommand{\row}{\text{\emph{row}}\xspace}
\newcommand{\far}{\text{\emph{Far}}\xspace}
\newcommand{\sw}{\text{\sc{sw}\xspace}\xspace}
\newcommand{\Gst}{G^{st}}

\newcommand{\near}{\text{\emph{Near}}\xspace}

\newcommand{\algn}{\alg({\near})\xspace}

\newcommand{\algcon}{\alg}

\newcommand{\Hl}{\ell_h}
\newcommand{\vl}{\ell_v}
\newcommand{\ppp}[2]{\text{\rm req}(#1\!\!\rightarrow\!\!#2)}

\ifnum\spaa=0
\title{\textbf{
Better Online Deterministic Packet Routing on Grids}}
\else
\title{
Better Online Deterministic Packet Routing on Grids}
\subtitle{[Regular Paper]}
\fi

\ifnum\spaa=0

 \author{Guy Even \thanks{School of Electrical Engineering, Tel Aviv Univ., Tel-Aviv 69978, Israel. ({\tt guy@eng.tau.ac.il}).}
 \and
 Moti Medina
 \thanks{LIAFA, Universit\'{e} Paris Diderot, 75205 Paris Cedex 13, France. ({\tt moti.medina@liafa.univ-paris-diderot.fr}).
 }
 \and
 Boaz Patt-Shamir\thanks{School of Electrical Engineering, Tel Aviv Univ., Tel-Aviv 69978, Israel. ({\tt boaz@tau.ac.il}).}
 ~\protect\footnote{Supported in part by ISF, MoST and Neptune.}
 }
\else
\numberofauthors{3}
\author{
\alignauthor Guy Even\\
       \affaddr{School of Electrical Engineering}\\
       \affaddr{Tel Aviv University}\\
       \affaddr{Tel Aviv 6997801, Israel}\\
       \email{guy@eng.tau.ac.il}
\alignauthor Moti Medina\\
       \affaddr{LIAFA}\\
       \affaddr{Universit\'{e} Paris Diderot}\\
       \affaddr{75205 Paris Cedex 13, France}\\
       \email{moti.medina@liafa.univ-paris-diderot.fr}
\alignauthor Boaz Patt-Shamir\thanks{Supported in part by ISF, MoST and
    Neptune.}\\
       \affaddr{School of Electrical Engineering}\\
       \affaddr{Tel Aviv University}\\
       \affaddr{Tel Aviv 6997801, Israel}\\
       \email{boaz@tau.ac.il}
}
\fi
\begin{document} 

\if\spaa=0
\def\thepage{}
\begin{titlepage}
\fi

 \maketitle
\begin{abstract}
  We consider the following fundamental routing problem.
  An adversary inputs packets arbitrarily at sources,
  each packet with an arbitrary destination. Traffic is
  constrained by link capacities and buffer sizes, and
  packets may be dropped at any time. The goal of the
  routing algorithm is to maximize throughput, i.e.,
  route as many packets as possible to their destination.
  Our main result is an $O\left(\log
  n\right)$-competitive deterministic algorithm for an
$n$-node line network (i.e., $1$-dimensional grid),
requiring only that buffers can store at least $5$
packets, and that links can deliver at least $5$ packets
per step. We note that $O(\log n)$ is the best ratio
known, even for randomized
  algorithms, even when allowed large buffers and wide
  links. The best previous deterministic algorithm for
  this problem with constant-size buffers and
  constant-capacity links was $O(\log^5 n)$-competitive.
  Our algorithm works like admission-control algorithms
  in the sense that if a packet is not dropped
  immediately upon arrival, then it is ``accepted'' and
  guaranteed to be delivered.
We also show how to extend our algorithm to a polylog-competitive
algorithm for any constant-dimension grid.
\end{abstract}

\ifnum\spaa=1
\keywords{
\else
\paragraph{Keywords.}
\fi
Online Algorithms, Packet Routing,  Bounded Buffers, Admission Control, Grid Networks
\ifnum\spaa=1
}
\fi


\if\spaa=0
\end{titlepage}
\pagenumbering{arabic}
\fi
\section{Introduction}
The core function of any packet-switching network is to route packets
from their origins to their destinations, but many
fundamental questions about packet routing are far from being well
understood. In this paper we consider one of these
questions, namely the \emph{competitive throughput network model}, introduced
by~\cite{AKOR}.
%
%
%

Briefly, the model is as follows.  The network consists of $n$ nodes
(switches) connected by point-to-point unidirectional communication
links, and we are given
two positive integer parameters, $B$ and $c$, called the buffer size
and link capacity, respectively.
%
Executions proceed as follows.
Packets are input by an adversary over time. Each packet is input at its \emph{source} node with a
given \emph{destination} node.
At each step, each packet is either forwarded over an incident link,
stored in its current location buffer, or dropped (i.e., removed from
the system).
Storing and forwarding are subject to the constraints that a buffer
can store at most $B$ packets simultaneously and that a link can carry
at most $c$ packets in a time step.  These constraints can be met for
all input sequences since the model allows for packets to be
{dropped} at any time.  The \emph{routing algorithm} selects, at
each step, which packets are forwarded, which are stored, and which
are dropped.  The goal of the algorithm is to maximize the number of
packets delivered at their destination.  Since we consider on-line
algorithms, we evaluate algorithms by their competitive ratio, i.e.,
the minimum ratio, over all finite packet input sequences, between the
number of packets delivered by the on-line algorithm and the maximum
number of packets that can be delivered by any (off-line)
constraint-respecting schedule.

It is nearly an embarrassment to find that very little is known about
this problem, even in the simplest case, where the network topology is
the trivial $n$-node unidirectional line. In this work we provide
an improved deterministic algorithm  for networks whose topology is
a $d$-dimensional grid.

\renewcommand{\arraystretch}{1.3}
\begin{table*}
\ifnum\spaa=0  \scriptsize \fi
\begin{center}
\begin{tabular}{|c|c|c|c|c|l|}
\hline
Ref. & Dim.& Comp.\ Ratio  & Deterministic?& Range of $B,c$& Remarks \tabularnewline
\hline
\hline
\cite{DBLP:conf/icalp/EvenM10,DBLP:conf/spaa/EvenM11,EM14}&1&$O(\log n)$&\checkmark&$B,c> \log n,~B/c = n^{O(1)}$&immediate from s-t reduction\\
\cite{AKK} & $1$ & $O(\log^3 n)$  &---  & $B\geq 2, c=1$ &  
\tabularnewline
\cite{AZ}   & $1$ & $O(\log^2 n)$ & --- &  $B \geq 2, c=1$ &  
FIFO buffers\tabularnewline
\cite{DBLP:conf/icalp/EvenM10,EM14}  & $1$ & $O(\log n)$ & --- &  $B
\in [1, \log n], c\geq 1$ &  
{also for $\log n \leq B/c \leq n^{O(1)}$}
\tabularnewline
\cite{DBLP:conf/spaa/EvenM11,EM14}   & $1$ & $O(\log^5 n)$ &
\checkmark &  $[3, O(\log n)]$ &  preemptive 
\tabularnewline
\cite{DBLP:conf/spaa/EvenM11,EM14}   &$d$ & $O(\log^{d+4} n)$ &
\checkmark &  $[3, O(\log n)]$ &  preemptive
\tabularnewline
\hline
\theoremref{thm:main}
&$1$ & $O(\log n)$ & \checkmark &  $[5, O(\log n)]$  & 
\tabularnewline
\theoremref{thm:d dim}&$d$ & $2^{O(d)}\cdot \log^d n$ & \checkmark &  $[2^{d+1}+1, O(\log n)]$  & 
\tabularnewline
\hline
\end{tabular}
\end{center}
\caption{\it Some results for centralized online algorithms for packet routing. The networks are
  uni-directional grids. In the special case
  of $B=0$ and $c \geq 3$, the algorithm in \cite{DBLP:conf/spaa/EvenM11,EM14}
  is $O(\log^{d+2} n)$-competitive.
}
\label{table:previous work}
\end{table*}
\renewcommand{\arraystretch}{1}

\subsection{Our Results}
Our main result is a centralized 
\emph{deterministic} $O(\log n)$-competitive packet routing algorithm for
unidirectional lines with $n$ nodes. 
The algorithm requires buffer size $B\geq 5$ and link capacity $c\geq 5$. In
addition, both $B$ and $c$ must be $O(\log n)$. We show how to extend the
algorithm to $d$ dimensions, where the competitive ratio is $2^{O(d)}\cdot \log^d n$,
assuming that $B,c\geq 2^{d+1}+1$. Our algorithm is \emph{nonpreemptive}, namely,
packets are dropped only at the time of their arrival (similarly to admission control
policies, which ``accept'' or ``reject'' requests upon arrival). By contrast,
{preemptive} algorithms may drop packets at any time, i.e., packets are not
guaranteed to reach the destination even after they start traversing the
network. 
The best previous deterministic algorithm \cite{DBLP:conf/spaa/EvenM11,EM14} is
preemptive.

\tableref{table:previous work} provides  a summary of our
results and a comparison of our algorithm with some
previous results along various aspects.

\newpage
\subsection{Overview of Techniques}
\label{sec:tech}

We first explain
our approach  for the 1-dimensional case.

The high-level idea is to reduce packet routing in a graph
$G$ to circuit switching (or path packing,
see~\cite{KT,AAP}) in the \emph{space-time graph}  $G\times
T$, where $T$ denotes the set of time steps. This so-called
{space-time transformation} has been used extensively in
this context~\cite{AAF,ARSU,AZ,RR,
DBLP:conf/icalp/EvenM10,DBLP:conf/spaa/EvenM11,EM14}.
To be effective, the space-time
transformation requires an upper bound on path lengths which does not
result in losing too
much throughput. We use the bound of \cite{DBLP:conf/spaa/EvenM11, EM14}
(which extends \cite{AZ}),
that ensures that the loss is at most some constant fraction.
After the transformation, we have an instance of
\emph{online path packing}~\cite{AAP,
BN06}. 
It is known that if the capacities are
large enough, i.e., $\log n\le B,c\le n^{O(1)}$, then online path
packing is solvable with logarithmic competitive
ratio~\cite{AAP,DBLP:conf/icalp/EvenM10,DBLP:conf/spaa/EvenM11,EM14}. 
We overcome the difficulty that $B$ and $c$ are $O(\log n)$ by employing a technique
called \emph{tiling}, i.e., partitioning the network nodes into large enough subgrids.
Tiling has been used in the past
\cite{KT,BL,DBLP:conf/icalp/EvenM10,DBLP:conf/spaa/EvenM11,EM14}; in
our algorithm, we use $4$ distinct tilings, 
and work on each of them independently.  Each tiling induces a new graph called the
\emph{sketch graph} whose nodes are the tiles. The capacity of the edges in the
space-time graph between adjacent tiles is $O(\log n)$ to allow for applying $O(\log
n)$-competitive path packing algorithms.  Path packing algorithms over the sketch
graph produce sketch paths for accepted packets. Thus, after these preliminary
simplifications, we arrive at the sub-task of \emph{detailed routing}, in which
coarse sketch  paths must be expanded to paths in the original space-time graph.

\emph{Fractional Optimum.}
Key to our application of the path-packing algorithm is
the analysis of Buchbinder and Naor~\cite{BN06,BNsurvey}, which
bounds the performance of the algorithm w.r.t.\ the \emph{fractional}
optimum, which may deliver packet fractions.
This result allows us to scale buffer sizes and link
capacities up and down while keeping the competitive ratio
under control.

\emph{Combining algorithms.} Another central component in the analysis of our
algorithm is the combination technique introduced by Kleinberg and Tardos \cite{KT}.
Loosely speaking, this technique deals with an admission control algorithm that is
the conjunction of two competitive algorithms, the state of which depends only on the
requests accepted by both. The technique enables one to prove that the competitive
ratio of the combined algorithm is the sum (rather than the product) of the
competitive ratios of the constituent algorithms.

\subsection{Previous Work}
Algorithms for dynamic routing on networks with bounded
buffers have been studied extensively both in theory and in practice
(see, e.g.,~\cite{AKRR} and references therein).
Let us first focus on centralized algorithms for $d$-dimensional grids.
%
We note that while centralized algorithms for packet routing were
always relevant for switch scheduling, recently the idea of centralization of
network functions, including route computation, gained substantial
additional traction due to
the concept of software-defined networks (SDN). See, e.g.,~\cite{nox}.
 The special case of $2$-dimensional grids (with or without
buffers) is of particular interest as this is the underlying
topology of crossbars in switches \cite{T}.

\emph{Online Algorithms for Unidirectional Lines.}
There is a series of papers on uni-directional line
networks, starting with~\cite{AKOR}, which introduced the
model. In~\cite{AKOR}, a lower bound of $\Omega(\sqrt{n})$
was proved for the greedy algorithm on unidirectional lines
if the buffer size $B\ge2$.  For the case $B=1$ (in a
slightly different model), an $\Omega(n)$ lower bound for
any deterministic algorithm was proved by~\cite{AZ,AKK}.
Both~\cite{AZ} and~\cite{AKK} developed, among other
things, online randomized centralized algorithms for
uni-directional lines with $B\geq 2$. In~\cite{AKK} an
$O(\log^3 n)$-competitive randomized centralized algorithm
was presented for $B\ge2$.  In addition,  it is proved in \cite{AKK}
that nearest-to-go is
$\tilde{O}(\sqrt{n})$-competitive for $B\geq 2$.  For the case $B=1$,
~\cite{AKK} presented a randomized
$\tilde{O}(\sqrt{n})$-competitive distributed algorithm.
(This algorithm also applies to rooted trees when all
packet are destined at the root.) In~\cite{AZ}, an
$O(\log^2 n)$-competitive randomized algorithm was
presented for the case $B\geq 2$.  (This algorithm also
applies to rings and trees.)
In~\cite{DBLP:conf/icalp/EvenM10}, an $O(\log n)$-competitive,
nonpreemptive, %
randomized algorithm was presented. The algorithm
in~\cite{DBLP:conf/icalp/EvenM10} is applicable to a wide range of
buffer sizes and link capacities, 
including the case $B=c=1$.  In~\cite{DBLP:conf/spaa/EvenM11}, an $O(\log^5
n)$-competitive 
deterministic algorithm was
presented. The algorithm in~\cite{DBLP:conf/spaa/EvenM11} is
applicable 
for $B,c \in [3,\log n]$.

\sloppy
\emph{Online Algorithms for Unidirectional 
Grids.} Angelov et al.~\cite{AKK} showed that the
competitive ratio of greedy algorithms in unidirectional
$2$-dimensional grids is $\Omega(\sqrt{n})$ and that
nearest-to-go policy achieves a competitive ratio of
$\tilde{\Theta}(n^{2/3})$.
In~\cite{DBLP:conf/spaa/EvenM11}, an $O(\log^6
n)$-competitive deterministic 
algorithm was
presented.
An extension of
this algorithm to $d$-dimensional unidirectional grids, with competitive
ratio $O(\log^{d+4} n)$,  is
presented in~\cite{DBLP:conf/spaa/EvenM11}.

For more related results, 
refer to~\cite{EM14}.

\paragraph{Organization\ifnum\spaa=0.\fi}
The problem is formalized in \sectionref{sec:problem}.
In \sectionref{sec:prelim} we explain the reduction of
packet-routing to path packing, and the construction of sketch graph.
%
In \sectionref{sec:alg} we describe the overall algorithm,
and in \sectionref{sec:analysis} we analyze it. Sections
\ref{sec:prelim}--\ref{sec:analysis} deal with the
$1$-dimensional grid (line); extension to the $d$
dimensional case is also discussed in
\sectionref{sec:d dim}.


\section{Model and Problem Statement}
\label{sec:problem}
\label{sect:problem}

We consider the standard model of synchronous store-and-forward packet
routing networks~\cite{AKOR,AKK,AZ}. The network is modeled by a
directed graph $G=(V,E)$, and by two integer parameter $B,c>0$.
For the most part of this paper, we consider a network whose topology
is  a \emph{directed line} of $n$
vertices, i.e.,
$V=\{v_0,\ldots,v_{n-1}\}$, $E=\Set{(v_{i-1},v_i\mid 0<i<n}$.

Execution proceeds in discrete steps. In step $t$, an arbitrary set of
\emph{requests} is input to the algorithm. Each request represents a
packet, and we will use both terms interchangeably. A request is
specified by a $3$-tuple
$r_i=(a_i,b_i,t_i)$, where $a_i\in V$ is the \emph{source node} of the
packet, $b_i\in V$ is its \emph{destination node}, and $t_i\in \mathbb{N}$ is
the time step in which the request is input.

In each time step, the \emph{routing algorithm} removes packets that
reached their destination, and  decides, for each
packet currently in the network, including
packets input in the current
step, whether
\begin{inparaenum}[(i)]
\item to drop the packet, or
\item to send it over an incident link, or
\item to store it in the current node.
\end{inparaenum}
The selection of the action is done subject to the following
considerations.
\begin{compactitem}
\item If a packet is dropped, it is lost forever.
\item A packet sent from node $u$ over link $(u,v)$ at time $t$ will
  be located at node $v$ at time $t+1$. The \emph{link capacity constraint} asserts that at any step,
  at most $c$ packets can be sent over each link.
\item A packet stored at node $u$ at time $t$ will
  be located at node $u$ at time $t+1$. The \emph{buffer capacity
    constraint} asserts that at any step,
  at most $B$ packets can be stored in each buffer.
\end{compactitem}
We use the following terminology.  A packet
$r_i=(a_i,b_i,t_i)$ is said to be \emph{input} (or arrive)
at $a_i$ at time $t_i$. We say that $r_i$ is
\emph{rejected} if it is dropped at time $t_i$, otherwise
it is \emph{accepted}. (Our algorithm will guarantee that
all accepted packets arrive at their destination.)

Given a set of requests, the \emph{throughput} of a packet routing
algorithm is the number of
packets that are delivered to their destination.  We consider the
problem of maximizing the throughput of an online centralized
deterministic packet-routing algorithm.
By \emph{online} we mean that by time $t$, the algorithm received as
input only requests that have been input by time $t$.
By \emph{centralized} we mean that the algorithm receives
all
requests and controls all  packets currently in the system without delay.
By \emph{nonpreemptive} we mean that every accepted packet reaches its
destination.




\emph{Competitive Ratio.}
Let $\sigma$ denote an input sequence. Let $\alg$ denote a
packet-routing algorithm.  Let $\alg(\sigma)$ denote the 
throughput obtained by $\alg$ on input $\sigma$. 
Let $\opt(\sigma)$ denote the largest possible subset of
requests in $\sigma$ that can be delivered without
violating the capacity constraints. We say that an online
deterministic \alg\ is \emph{$\rho$-competitive}, if for
every input sequence $\sigma$, $|\alg(\sigma)| \geq \frac
1\rho \cdot |\opt(\sigma)|$. Our goal is to design an
algorithm with the smallest possible competitive ratio.

\section{First Steps}
\label{sec:prelim}
In this section we present preliminary simplifications we apply to the
problem. They include reducing the packet routing on a line problem to path
packing on grids, and then path packing on sketch graphs.

\subsection{From Packet-Routing on a Line to Path Packing in a Grid}
\label{sec:reduction}


Let $G=(V,E)$ denote a directed line with link capacities $c$ and
buffer sizes $B$. The space-time grid of $G=(V,E)$ is a directed
acyclic infinite graph $G^{st} =(V^{st},E^{st})$ with edge capacities
$c^{st}(e)$, where
\begin{inparaenum}[(i)]
\item $V^{st} \triangleq V\times \mathbb{N}$. Each vertex $v \in V$
  has infinitely many copies in the space-time grid $G^{st}$; namely,
  vertex $(v,t)\in V^{st}$ is the copy of $v$ that corresponds to time
  $t$.
\item $E^{st}\triangleq E_0\cup E_1$ where $E_0$ denotes forward edges
  and $E_1$ denotes the store edges.  Formally, $E_0\triangleq \{
  (u,t)\rightarrow(v,t+1)\::\: (u,v)\in E~,~t\in\mathbb{N}\}$ and $E_1
  \triangleq \{ (u,t)\rightarrow (u,t+1) \::\: u\in V, t\in
  \mathbb{N}\}$.
\item The capacity of all edges in $E_0$ is $c$, and all edges in
  $E_1$ have capacity $B$.
\end{inparaenum}

\emph{The transformation.} We transform a request
$r_i=(a_i,b_i,t_i)$ for routing a packet in the directed
line $G$ to a path request $r^{st}_i=((a_i,t_i),\row(b_i))$
in the grid $G^{st}$. The correctness of the reduction is
based on a one-to-one correspondence between paths in
$G^{st}$ and a routing of a packet in $G$. Each vertical
edge $(v_i,t)\rightarrow(v_{i+1},t+1)$ in $G^{st}$
corresponds to forwarding a packet from $v_i$ to $v_{i+1}$
in step $t$, and each horizontal edge
$(v_i,t)\rightarrow(v_{i},t+1)$ in $G^{st}$ corresponds to
storing a packet in $v_i$ in step $t$.

\emph{Embedding in the plane.} The na\"ive depiction of
$G^{st}$ maps vertex $(v_i,t)$ to the point $(t,i)$ in the
plane (i.e., the $x$-axis is the time axis and the $y$-axis
is the ``vertex-index'' axis). This embedding of $G^{st}$
results with a lattice of vertices in which edges are
either horizontal or diagonal.  We prefer the embedding in
which the edges are axis parallel, which means that vertex
$(v_i,t)$ is mapped to the point $(t-i,i)$.  In the
axis-parallel depiction, all the copies of a vertex $v_i\in
V$ still reside in the $i$th row. However, column $j$
corresponds to a traversal  of the complete line, starting
at $v_0$ at time $j$ and ending at $v_{n-1}$ at time
$j+n-1$.

\subsection{From One Grid to Four Sketch Graphs}\label{sec:sketch graph}
Given a grid generated by the transformation above, we apply another
transformation to produce a coarsened version, called the \emph{sketch
  graph}. Specifically, we use
\emph{tiling}. Tiling is a partition of the grid nodes into
 $\Hl\times \vl$ subgrids, where $\Hl$ and  $\vl$ are parameters to
 be determined later.  We also add dummy nodes to the space-time grid
$G^{st}$ to complete all tiles.  This augmentation has no effect on
routing because a dummy vertex
does not belong to any route between real vertices.

The tiling is specified by two additional parameters $\phi_{x}$ and $\phi_{y}$ called
\emph{offsets}.  The offsets determine the positions of the corners of the tiles;
namely, the left bottom corner of the tiles are located in the points
$(\phi_{x}+i\cdot\Hl ,\phi_{y}+j\cdot\vl)$, for $i,j\in \mathbb{N}$. The algorithm
uses four offsets $(\phi_{x}, \phi_{y})\in \{-\Hl/2,0\}\times \{-\vl/2,0\}$.  We
denote these four tilings by $T_1,\ldots,T_4$.

\begin{proposition}\label{prop:sw}
  For every vertex $(v,t)$ of the space-time grid $G^{st}$, there exists exactly one
  tiling $T_j$ such that $(v,t)$ is in the south-west quadrant of a tile of $T_j$.
\end{proposition}
\noindent \propref{prop:sw} suggests a partitioning of the requests.
\begin{definition}
  A request $r_i=(a_i,b_i,t_i)$ is in $\sw_j$ if the source vertex $(a_i,t_i)$ of request $r_i$
  belongs to the south-west quadrant of a tile in the tiling $T_j$.
\end{definition}

\emph{The Sketch Graphs.}  Each tiling $T_j$ induces a grid, called the sketch graph,
each vertex of which corresponds to a tile. The sketch graph induced by $T_j$ is
denoted by $S_j\DEF(V(S_j),E(S_j))$, where $V(S_j)$ is the set of tiles in $T_j$.
There is a directed edge $(s_1,s_2)\in E(S_j)$ if $s_1\neq s_2$ and $E^{st}\cap (s_1\times s_2)\neq
\emptyset$.  All edges in the sketch graph are assigned unit capacity.

\subsection{Online Packing of Paths}\label{sec:IPP}
We use  the sketch graphs to solve \emph{path packing} problems.
%
Intuitively, the path packing model resembles the packet routing
model, except that there are no buffers, and that each link $e$ may
have a different capacity $c(e)$.
In addition,  we
generalize the notion of a request to allow for a set of destinations
(similar to ``anycast'') as follows. Usually, the destination of a request consists of a
single vertex. If $G$ is a directed graph, then it is easy to reduce
the case in which the destination is a subset to the case in which
the destination is a specific vertex. The reduction simply adds a
sink node that is connected to every vertex in the destination
subset. In our setting of space-time grid, the destination subset is
a row. Thus it suffices to add a sink node for each row (as in~\cite{AZ}).

Formally, a \emph{path request} $r_i$ in $G$ is a pair $(a_i, D_i)$, where $a_i\in V$
is the source vertex and $D_i\subseteq V$ is the destination subset.  Let $P(r_i)$
denote the set of paths that can be used to serve request $r_i$; namely, every path
$p\in P(r_i)$ begins in $a_i$, ends in a vertex in $D_i$, and satisfies some
additional constraint (e.g., bounded length, bounded number of turns, etc.).  Given a sequence $R=\{r_i\}_{i\in I}$
of path requests, we call a sequence $P=\{p_i\}_{i \in J}$ a \emph{partial routing}
of $R$ if $J\subseteq I$ and $p_i\in P(r_i)$ for every $i\in J$.  The \emph{load} of
an edge $e\in E$ induced by $P$ is the ratio $|\{p_j \in P: e\in p_j\}|\over
c(e)$. 
A partial routing of a set of path requests is called a \emph{$\beta$-packing} if the
load induced on each edge is at most $\beta$.  The \emph{throughput} of $P$ is simply
the number $|J|$ of paths in $P$.

\paragraph{Integral and Fractional Partial Routings\ifnum\spaa=0.\fi}
In the integral scenario,
a path request is either  served by a single path or is not served.
In \emph{fractional routing},
a request $r_i$ can be (partially) served by a combination of paths
$p_1,\ldots,p_k$. Namely, each path $p_j$ serves a fraction
$\lambda_j$ of the request, where $\lambda_j\ge 0$ for all
$j$ and  $\sum_j \lambda_{j}\le1$. We refer to $\sum_j \lambda_{j}$ as the \emph{flow
amount} of request $r_i$. The \emph{load} of an edge $e\in
E$ induced by request $r_i$ is the ratio $\sum_{j: e\in
p_j} \lambda_{j} / c(e)$. A fractional solution is
$\beta$-packing if the total load on in each edge, from all
requests, is at most $\beta$.   The \emph{throughput} of a
fractional routing is the sum of the flow amounts of all requests. Given a
fractional routing $g$, we use $|g|$ to denote its
throughput.
Trivially, the maximum throughput attainable
by a fractional $\beta$-packing is an upper bound on the maximum
throughput attainable by an integral $\beta$-packing.  An
optimal-throughput fractional $\beta$-packing can be computed off-line
by
solving a linear program.

\paragraph{Online Path Packing: Problem and Solution\ifnum\spaa=0.\fi}
In the
online path packing problem, the input is a sequence of
path requests $R=\{r_i\}_{i\in I}$.  Upon arrival of a
request $r_i$, the algorithm must either allocate a path
$p\in P(r_i)$ to $r_i$ or reject $r_i$.  An online path
packing algorithm is said to be
\emph{$(\alpha,\beta)$-competitive} if it computes a
$\beta$-packing whose throughput is at least $1/\alpha$
times the maximum throughput over all $1$-packings. Note
that for online path packing, we assume that all edges have
capacity at least $1$.

The online path packing algorithm in~\cite{AAP} (analyzed also by~\cite{BN06})
assigns weights to the edges that are exponential in the load of the edges. This load
is the load incurred by the paths allocated to the requests that have been accepted
so far. The algorithm is based on an oracle that is input $r_i$ and the edge weights,
and outputs a lightest path $p_i$ in $P(r_i)$. If the weight of $p_i$ is large, then
request $r_i$ is rejected; otherwise, request $r_i$ is routed along $p_i$.
We refer to the online algorithm for online integral path packing by
$\route$.
The competitive ratio of the \route\ algorithm is summarized in the following theorem.
\begin{theorem}[\cite{EM14}, following~\cite{AAP, BN06}]
\label{thm:IPP}\sloppy Consider an online path packing
problem on an infinite graph with edge capacities such that $\inf_{e}c(e) \geq 1$.
Assume that, for every request $r_i$, the length of every legal path in $P(r_i)$ is
bounded by $p_{\max}$.  Then algorithm $\route$ is $(2,\log(1+ 3\cdot
\pmax))$-competitive online integral path packing algorithm.  Moreover, the
throughput of $\route$ for any request sequence is at least $1/2$ the throughput of
any fractional packing for that sequence.
\end{theorem}

\paragraph{Bounded Path Lengths\ifnum\spaa=0.\fi}
The load obtained by the \IPP\ algorithm is logarithmic in the maximum path length
$\pmax$. This suggests that $\pmax$ should be polynomial in $n$.
Lemma~\ref{lemma:nB} states that limiting the number of store steps per packet by a
polynomial in $n$ decreases the fractional throughput only by a constant factor.

We use the following notation. Given a request sequence $R=\{r_i\}_i$, let
$f^*(R)$ denote a maximum throughput fractional $1$-packing of $R$, and let
$f^*(R|\pmax)$ denote a maximum throughput fractional $1$-packing with respect to
$R$ under the constraint that each path is of length at most $\pmax$.
\begin{lemma}[after \cite{AZ}]\label{lemma:nB}
  Let $\pmax\triangleq 2n \cdot (1+\frac{B}{c})$.  Then $|f^*(R| \pmax)|
  >0.31\cdot |f^*(R)|$.
\end{lemma}

\begin{figure}[t]
      \centering
        \includegraphics[width=0.35\textwidth]{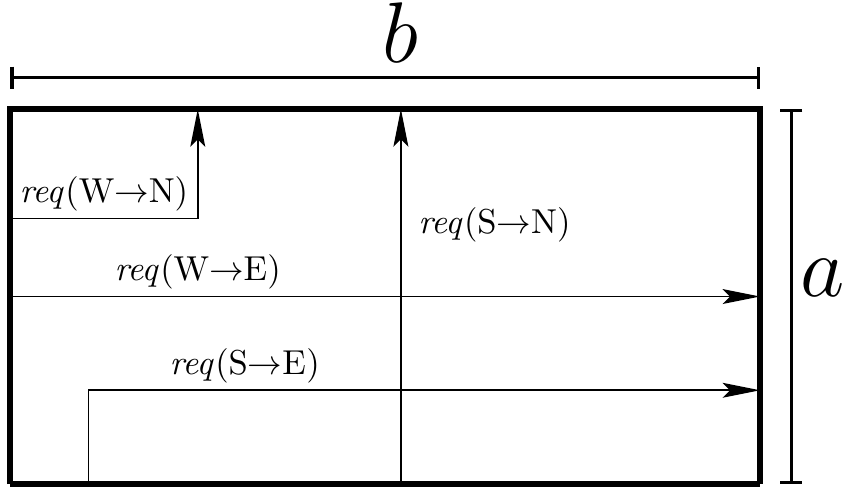}
        \caption{\em An $a \times b$  grid. The four types
        of requests (i.e., $\ppp{W}{N}, \ppp{S}{N}, \ppp{W}{E}, \ppp{S}{E}$) are depicted
        by the four arrows.}
      \label{fig:crossbar}
    \end{figure}

\subsection{Routing paths across 2-d
  Grids}
Consider the following special case of routing in grids.
Suppose that each path request has a specific source vertex which
resides on either the south of the west side, and the destination is
either the north or the east \emph{side} (i.e., we can route to any
vertex on the requested side).
  For $X\in\Set{S,W}$ and $Y\in\Set{N,E}$, let
$\ppp{X}{Y}$ denote the set of path requests whose source
is in the $X$ side and whose destination is the $Y$ side
(see \figref{fig:crossbar}).
The following claim establishes  sufficient and necessary conditions
for satisfying such path requests. We refer to the routing algorithm used in this
case as \emph{crossbar routing}.
\begin{proposition}[Crossbar routing]
\label{prop:crossbar}\sloppy
The path requests can be packed in the $2$-dimensional $a \times b$
directed grid  if and only if $|\ppp{W}{E}|+|\ppp{S}{E}|\leq a$ and
$|\ppp{W}{N}|+|\ppp{S}{N}|\leq b$.
\end{proposition}
\begin{proof}
  The ``only if'' part is obvious. We now present a distributed
  algorithm that proves the ``if'' part.  Without loss of generality,
  we may assume that $\ppp{S}{N}$ and $\ppp{W}{E}$ are empty. This
  assumption is satisfied by routing such requests along straight
  paths and giving them precedence over other requests. Thus we may
  ignore these lines henceforth, and we are left with the task of
  routing $\ppp{S}{E}$ and $\ppp{W}{N}$ under the assumption that
  $|\ppp{S}{E}|\leq a$ and $|\ppp{W}{N}|\leq b$.

  These requests are served as follows.  Order the rows from bottom
  to top and the columns from left to right. Assume, w.l.o.g.,  that $a\leq b$
  (the case that $a>b$ is solved analogously).

  Requests whose source vertex is in the first $a$ rows or columns
  turn in the vertex along the diagonal emanating from the SW
  corner. For example, a request in $\ppp{W}{N}$ whose source
  is in row $i$ is routed eastward for $i$ hops, and then
  north for $a-i$ hops (i.e., to the north side of the grid). See
  \figref{fig:crossbar2}.

  The requests whose source vertex is in the last $b-a$ columns are
  routed northward until they reach a vertex that does not receive an
  east-bound path from its west neighbor. Once such a vertex is found,
  the path turns east and continues straight until it reaches the east
  side of the grid.  Indeed, such a right turn is always possible
  because $a\geq |\ppp{S}{E}|$ and hence a ``vacant row'' is always
  found.
\end{proof}

\begin{figure}[t]
      \centering
        \includegraphics[width=0.35\textwidth]{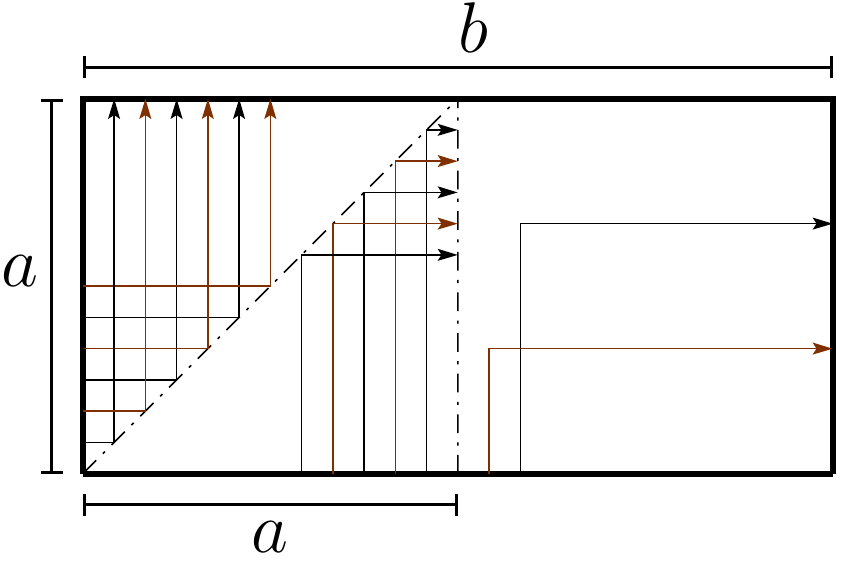}
        \caption{\em Satisfying the path requests in the $a \times b$ two
          dimensional directed grid, where $a\leq b$. 
        }
      \label{fig:crossbar2}
    \end{figure}

\begin{rem}
  \propref{prop:crossbar} extends to the case of capacitated edges
assuming all horizontal edges have the same capacity and all
  vertical edges have the same capacity.
In this case, the requests can be routed iff
  the number of requests for each destination side is bounded by
  total capacity of edges crossing that side.
\end{rem}
\section{The Packet Routing Algorithm}
\label{sec:alg}
\begin{algorithm}[t]
\ifnum\spaa=1
\small
\fi
\caption{Top-level algorithm for packet routing in the $1$-dimensional
  grid. Code for step $t$.}
\label{alg:outline}
\begin{algorithmic}[1]
\State Let $R_t$ be a list of new requests, sorted by
source-destination distance. \label{l0}
\State For each vertex $v$, let $R'_t(v)$ the first $B'+c'$ requests in $R_t$
whose source is $v$.\Comment{filter requests}
\label{line:filter}
  \For{each request $r_i 
    \in\bigcup_v R'_t(v)$}\label{l00}

\If {$r_i \in \near$}{ $\routenear (r_i)$} \label{l1}
\Else
\State Let $j\in\Set{1,\ldots,4}$ be s.t.\  $r_i\in \sw_j$\Comment{classify $r_i$}\label{l2}
\State \label{item:IPP} $sketch_i \gets \IPP(S_j,accepted_j, r_i)$
\label{l5}\Comment{path lengths bounded by $\pmax$}
\State \label{item:init-route} $init_i \gets \initroute(accepted_j, r_i)$ \label{l6}

\If {$sketch_i \neq \text{REJECT}$ \text{ and } $init_i \neq \text{REJECT}$}

\State add $r_i$ to $accepted_j$
\State $\detailedroute(r_i,init_i,sketch_i)$\Comment{update  routes}

\Else {~Reject $r_i$}
\EndIf
\EndIf
\EndFor
\end{algorithmic}
\end{algorithm}

We now present the routing algorithm. Pseudo-code is provided in
\Algref{alg:outline}.  The algorithm works as follows. First, in lines
\ref{l0}-\ref{l00}, an initial filtering of the requests removes requests if too many
requests originate in the same space-time vertex (see~\defref{def:R'}). Then each
remaining new request is processed. In lines \ref{l1}-\ref{l2}, it is classified as
either \near or \far, based on its source-destination distance (see paragraph on
packet classification).  \near requests are routed by the \routenear\ algorithm,
described in \sectionref{sec:near}.
Each \far request is associated with the tiling
$T_j$ in which its source vertex belongs to a south-west quadrant of a tile. Each
tiling is processed separately by three procedures: (i)~The \IPP\ algorithm, which
performs online path packing over the sketch graph $S_j$ (line~\ref{l5}). The outcome
$sketch_i$ is either ``REJECT'' or a path in a sketch graph $S_j$, i.e., a sequence
of tiles from the initial tile to the destination tile.  (ii)~The $\initroute$
procedure looks for a routing within the SW-quadrant of the first tile of $r_i$: its
outcome is either such a path denoted $init_i$ or ``REJECT''. Only $\IPP$ and
$\initroute$ may reject a far request. If both procedures are
successful, then $\detailedroute$ is called (line~\ref{l6}). Detailed routing computes a path in
the space-time graph, i.e., a complete schedule for each packet.  In our algorithm,
the sketch path for each accepted request is computed once and it is fixed, but
the future part of a detailed route of a request may change due to the insertion of new packets.
Therefore, the procedure $\detailedroute$ not only computes a path for $r_i$ in
$G^{st}$, but may also alter the detailed routes of other requests (without changing
the high-level sketch-graph routes).

An important property of $\IPP$ and $\initroute$ is that their state
is determined by the requests that are actually in the system, i.e.,
accepted by both.  (Rejected requests by either do not affect the
state of the system.)  This property enables us to employ the
combination technique of~\cite{KT}.  The listing emphasizes this
property by explicitly managing the sets of accepted requests for
each class (denoted by $accepted_j$). These sets are arguments of
$\IPP$ and $\initroute$ and determine their states. We now proceed to
explain the algorithm in detail.

\paragraph{Packet classification\ifnum\spaa=0.\fi}
A request $r_i=(a_i,b_i,t_i)$ is called \emph{near} if $b_i-a_i \leq \vl$, and
\emph{far} otherwise. We denote the sets of near and far requests by \near\ and \far,
respectively. The far requests are further classified into four classes denoted by
$\far_j$, where $\far_j\DEF \far \cap \sw_j$. Namely, $\far_j$ is the set of far
requests whose source node is in the SW-quadrant of a tile $s$ in the tiling $T_j$.

\paragraph{Link Multiplexing\ifnum\spaa=0.\fi}
After classification, there are five classes of
requests: one for \near and one for each tiling. Routes for each class
are computed independently and the final result is the union of routes
over all classes. This is made possible by partitioning the capacity
of each edge in the space-time
grid into virtual \emph{tracks}, one track per class.  The capacity of each
track is $\frac15$ of the capacity of the edge, rounded
down, i.e., track capacities are $B'\triangleq \lfloor B/5
\rfloor$ and
{$c'\triangleq \lfloor c/5 \rfloor$}.  (This explains why
we require that $B,c\geq 5$.)

\paragraph{Tiling Parameters\ifnum\spaa=0.\fi}
Tile side lengths are set so that the trivial greedy routing algorithm
is $O(\log n)$-competitive for requests that can be satisfied within a
tile. Each tile
has length $\Hl$ and height $\vl$, defined as follows.
Recall that the maximum path length $\pmax= 2n  \cdot (1+\frac{B}{c})$ (cf.~\lemmaref{lemma:nB}).
\begin{defn}\label{def:xy}
We use the following parameters.
\begin{compactitem}
\item $k\triangleq\log(1+ 3\cdot \pmax)$
\item $\Hl=  \lceil\frac{6k}{5c'}\rceil$ and $\vl=\lceil\frac{6k}{5B'}\rceil$
\end{compactitem}
\end{defn}
\noindent
We summarize with the following claim.

\begin{proposition}\label{prop:tiling}
If $B/c$ is bounded by a polynomial in $n$, then the tiling parameters
satisfy the following properties.
  \begin{compactenum}
  \item $\Hl+\vl = O(\log n)$.
  \item The sum of the edge capacities along each tile side is $\Theta(k)$.
  \item For each track, the sum of the track capacities along a tile side is at least $6k$.
 \end{compactenum}
\end{proposition}

\begin{proof}
  Clearly $\Hl+\vl=O(k)$. If $B/c$ is polynomial in $n$, then
  $k=O(\log n)$.  The sum of the edge capacities along a vertical side
  is $\vl\cdot B = \Theta(k)$.  The sum of the track capacities
  crossing a vertical side of a tile is at least $\vl \cdot 5B' \geq
  6k$. The capacities along a horizontal edge is bounded similarly.
\end{proof}

\paragraph{Filtering superfluous simultaneous requests with identical sources\ifnum\spaa=0.\fi}
Since we do not impose any restriction on the requests, it could well be that many
requests arrive at the same source vertex in a single time step. To deal with that,
we use that fact that for each node $v$ and step $t$, no more than $c+B$ requests can
leave $(v,t)$ in \emph{any} routing. The partition of link capacities for tracks
imposes a stricter limitation in the sense that within each class, no more than
$c'+B'$ paths can have the same source vertex.
\begin{definition}\label{def:R'}
  Given a sequence $R$ of requests, let $R'$ denote the subsequence of $R$ defined as
  follows.  For each source vertex $(a_i,t_i)$, choose $c'+B'$
  packets whose destination is closest to the source node. (If at most $B'+c'$
  requests originate at the same node, then all of them are kept
  in $R'$.)
\end{definition}
\noindent
\propref{prop:R'} shows that rejecting the requests in $R\setminus R'$ reduces the
fractional optimal throughput only by a constant factor.
\subsection{Routing Rules}
The routing at the high level (sketch path) is determined by the
$\IPP$ algorithm.
We now explain the
ideas behind refining these rough paths (in the sketch graph) into
actual paths (in the space-time grid). Throughout this section we
consider, w.l.o.g.,  a single tiling $T_j$.

Fix a tile $s$ IN $t_J$. We distinguish between the following three types of requests in
$\far_j$ (we deal with \near requests in \sectionref{sec:near}).
\begin{compactitem}
\item \emph{Initial requests:} requests whose source vertex is in the
  south-west
  quadrant of the tile $s$.
\item \emph{Traversing requests:} these are requests that enter $s$ from a
  specific vertex on one side (either west or south) and must leave
  through any vertex of another side (either east or north).  The
  entry vertex is determined by a previously-invoked detailed routing,
  and the exit side is determined by the sketch path.
\item \emph{Final requests:} these are requests whose sketch path ends in
  tile $s$.  The destination of a final request is the north side of
  $s$.%
  \footnote{The detailed path for $r_i$ should end in a copy of the
    destination vertex $b_i$ in $s$. As any path that reaches the
    north side of $s$ also reaches a copy of $b_i$, we use the
    pessimistic assumption that $\row(b_i)$ is the northmost row of
    tile $s$.}
\end{compactitem}

    \begin{figure}
      \centering
        \includegraphics[width=0.35\textwidth]{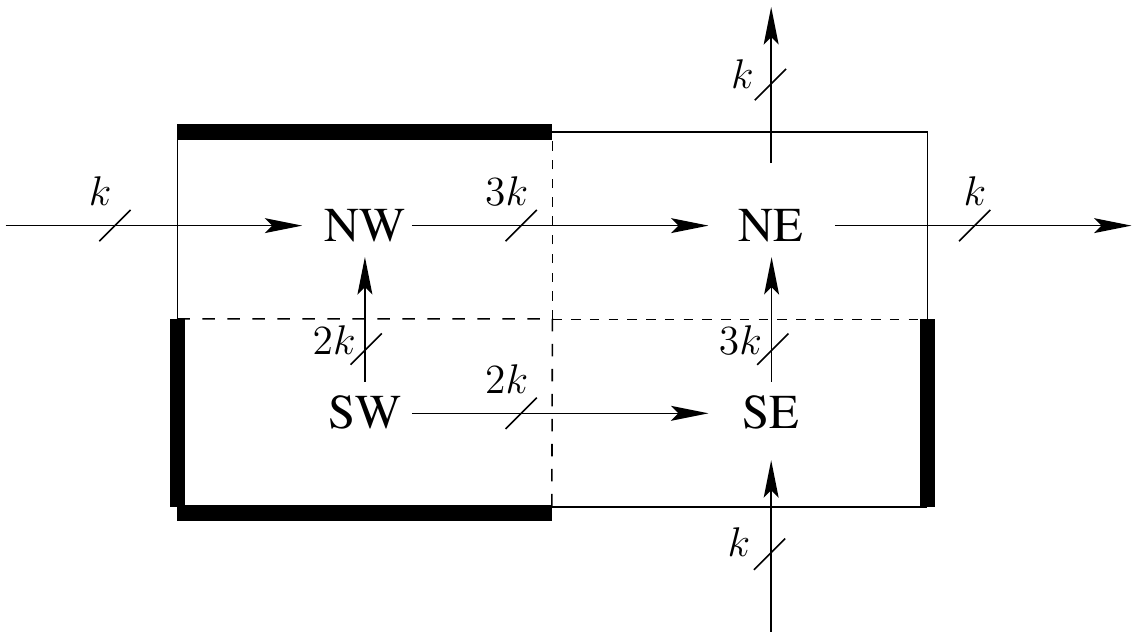}
        \caption{The quadrants of a tile. Traversing requests may not cross the thick
          border lines. Upper bounds on the number of paths in each class that
          traverse the quadrant sides follow from the path packing algorithm \IPP. }
      \label{fig:detail}
    \end{figure}

Each tile is partitioned into $4$ quadrants,
denoted NE, SE, SW and NW.  We constrain the way requests are routed
within a tile
using the following rules (see \figureref{fig:detail}; no request
may cross a thick line).

\begin{compactenum}
\item Initial requests always start in the SW-quadrant and are routed to the north or
  east side of the SW-quadrant along a straight path.  The SW-quadrant
  of each tile is reserved for routing of initial requests.
\item Traversing requests whose source and destination sides are
  opposite (e.g., from the south to the north side) are routed along a
  straight path.
\item Paths enter the tile either in the
  east half of the
  south side or the north half of the west side.
\item Paths exit the tile through either the east half of the north side or the
  north half of the east side.
\end{compactenum}

\subsection{Procedure \initroute}
Initial routing takes place in the SW-quadrant of the first tile of a
\far request.  The goal of \initroute is modest:  route the request
to the boundary of the SW-quadrant. This is done greedily  along
straight paths if possible: it may be the case that some of the edges
are already reserved for another request (also routed by \initroute of
the same tile). If no straight path is available, \initroute returns REJECT.
 This means that incoming traffic (of earlier requests)
continues uninterrupted along a straight path. Only remaining capacity
along edges that emanate from a vertex $(i,j)$, if any, is used for
routing the requests that originate in $(i,j)$. 

\subsection{Procedure \detailedroute}
The goal in detailed routing is to compute a detailed path $p_i$ in
the space-time graph $G^{st}$ given a sketch path $sketch_i$ in the
sketch graph $S_j$ and the initial part of the route $init_i$.  The
sketch path specifies the sequence of tiles to be traversed by the
detailed path.  In addition, the sketch path specifies the tile sides
through which the detailed path should enter and exit each tile. 
Requests that have been
assigned a sketch path and an initial route must be successfully
routed by detailed routing.

Detailed routing is computed by applying crossbar routing (cf.\
\propref{prop:crossbar}) to the NW, SE and NE quadrants. This routing is computed
based on the present requests.  As new requests arrive, the future portions of the
detailed routes may change dynamically so that all requests which are ``in
progress'' will reach their destination. 
Below we argue that crossbar routing indeed succeeds.
\begin{claim}
Detailed routing successfully completes the route of each accepted far request.
\end{claim}
\begin{proof}
By \propref{prop:crossbar}, to ensure successful routing
it is sufficient to bound the number of paths
that need to traverse a quadrant by the capacity of the
quadrant side.
  By~\propref{prop:tiling}, the track capacity of each quadrant side
  is at least $3k$.  We now prove upper bounds on the number of paths
  that traverse each quadrant side (see~\figref{fig:detail}).  The
  \IPP\ path packing algorithm is a $k$-packing over the sketch graph
  (whose edges have unit capacity). It follows that at most $k$ paths
  traverse each side of the tile.  As every request that originates in
  the SW-quadrant of a tile must exit the tile, there are at most $2k$
  paths that traverse each side of the SW-quadrant (although their sum
  is also bounded by $2k$).  Hence the upper bounds depicted
  in~\figref{fig:detail} follow. We need to elaborate more on the
  NE-quadrant because it is also used for routing final requests
  (i.e., requests that do not exit the tile, but do want to reach its
  top row). Consider the north side of the NE-quadrant. There are at
  most $k$ traversing requests that wish to exit the tile. In
  addition, there are at most $2k$ final requests that wish to reach
  to top row (as each final far request must have entered the
  tile). Thus, in total there are at most $3k$ paths that wish to
  reach the top side of the NE-quadrant. To summarize, the number of
  paths that wish to reach any quadrant side is bounded by the side's
  capacity, and hence by~\propref{prop:crossbar}, detailed routing
  succeeds.
\end{proof}

Finally, we note that in order for \detailedroute\ to be well defined, we compute it
in tiles in column-major order, i.e., we start with the bottom tile of the leftmost
row and go up, then the bottom tile of the second-from left column and go up etc.
This ensures that when we reach a tile, all input vertices are fixed. We remark that
detailed routing can be executed in a local distributed manner; in each time step,
each vertex needs only to know the initial paths the sketch paths of the incoming
packets.

\subsection{Procedure \routenear}\label{sec:near}

Finally, we describe the algorithm for the near requests.  The
\routenear\ Algorithm is extremely simple: it never stores a packet
(i.e., it uses only vertical edges in $G^{st}$, and gives precedence
to older requests). In more detail, upon arrival of a request $r_i\in
\near$, the algorithm checks the number of requests already routed
along the outgoing vertical edge (from $(a_i,t)$ to $(a_i+1,t+1)$). If
this number is less than $c'$, then the algorithm routes $r_i$ along
the vertical path in $G^{st}$ from $(a_i,t)$ to
$(b_i,t+(b_i-a_i))$. Note that these edges occur in the future, and
hence cannot have been saturated by \routenear if the edge outgoing
from $(a_i,t)$ is not saturated.  If there is no free capacity in the
outgoing vertical edge, $r_i$ is rejected. Note that if $r_j$ is
accepted, then it guaranteed to reach its destination.

\section{Analysis of Competitive Ratio of the Routing Algorithm}
\label{sec:analysis}

Our goal is to prove the following theorem for a directed line $G$ of $n$ vertices with
buffer sizes $B$ and link capacities $c$, where $B,c\in[5,\log n]$.
\begin{theorem}\label{thm:main}
  \Algref{alg:outline} is $O(\log n)$-competitive with respect to the throughput of a
  maximum fractional routing.
\end{theorem}

We translate the problem to a path packing problem over the space-time graph $\Gst$.
Let $f^*_{\Gst}(R)$ denote a maximum throughput fractional  routing, and let
$|f^*_{\Gst}(R)|$ denote its throughput. Let $|\alg(R)|$ denote the throughput of the
online packet algorithm. \theoremref{thm:main} follows directly from the following
lemma.

\begin{lemma}\label{lemma:frac cr} For every sequence of requests $R$,
  $|f^*_{\Gst}(R)|\leq O(\log n) \cdot |\alg(R)|$.
\end{lemma}

We outline the proof of~\lemmaref{lemma:frac cr}.  We scale the capacities down by a
factor of $\Theta(k)=\Theta(\log n)$ in the sketch graph.  By linearity, this reduces
the optimal fractional throughput by the same factor (see~\propref{prop:scale}).  We
show that the filtering stage in~\lineref{line:filter} incurs only a constant factor
reduction to the optimal fractional throughput (see \propref{prop:R'}).  The filtered
requests $R'$ are partitioned into near requests and far requests (which are further
partitioned into $4$ classes, one per tiling). The far and near requests are analyzed
separately.  The analysis of the throughput for far requests builds on the competitive
ratio of the \IPP\ algorithm and the \initroute\ algorithm (see~\claimref{claim:IPP}
and~\claimref{claim:init}). By applying the combining analysis of Kleinberg and
Tardos~\cite{KT}, we show that the competitive ratio for the combined algorithm is
the sum of the two algorithms (see~\claimref{claim:combine}).
In~\theoremref{thm:near}, we show that the \routenear\ algorithm succeeds in routing
a logarithmic fraction of the filtered near requests.  In~\sectionref{sec:ptt}, the
parts of the proof are combined together to prove~\lemmaref{lemma:frac cr}.

\subsection{Scaling and Filtering}
One advantage of working with fractional routings is that, by linearity, the
throughput scales exactly with the capacities.  Let $f^*_{S_j}(R)$ denote a maximum
throughput fractional routing in the sketch graph $S_j$.  Recall that the sketch
graph has unit capacities. Coalescing of vertices of $G^st$ in each tile results with
edge capacities that are $\Theta(k)=\Theta(\log n)$. Hence, we obtain the following
proposition.
\begin{proposition}\label{prop:scale}
  $|f^*_{S_j}(R)| = \frac{1}{\Theta(k)} \cdot |f^*_{\Gst}(R)|$.
\end{proposition}

\noindent
Recall that by~\defref{def:R'}, in the input sequence $R'$, at most $B'+c'$ requests
originate in each space-time vertex.
\begin{proposition}\label{prop:R'}
$   |f^*_{\Gst}(R)| \leq 9 \cdot  |f^*_{\Gst}(R')|$.
\end{proposition}
\begin{proof}
  Fix a space-time vertex $v=(a_i,t_i)$.  Let $R_v$ (reps. $R'_v$) denote requests in
  $R$ (reps. $R'$) that originate in $v$. Let $f=f^*_{\Gst}(R)$. Consider the flow $f/9$.  For
  every vertex $v$, the amount of flow that originates in $v$ is bounded by $(B+c)/9
  \leq B'+c'$.  Divert flow from in $f/9$ from $R_v\setminus R'_v$ to $R'_v$ along shorter
  paths, to obtain a flow $g$ with respect to $R'$ such that $|g| = |f|/9$. Since
  $|f^*_{\Gst}(R') |\geq |g|$, the proposition follows.
\end{proof}

\subsection{Far Requests}
Two algorithms determine whether a far request is rejected: (i)~the \IPP\ path
packing algorithm over the sketch graph, and (ii)~the \initroute\ algorithm that
deals with routing in the initial SW-quadrant of the source tile. We begin by showing
that, if invoked separately, each of these algorithms accepts at least a constant
fraction of the maximum fractional throughout over the sketch graph.

Let $R'_j$ denote the subsequence of requests in $R'$ that are in the class $\far_j$.
Suppose we invoke the \IPP\ algorithm in isolation over the sketch graph $S_j$ with
the input sequence $R'_j$. By isolation we mean that the accepted requests are
determined solely by $\IPP$. Let $|\IPP_{S_j}(R'_j)|$ denote the number of requests
that are accepted by this invocation.
\begin{claim}\label{claim:IPP}
$|f^*_{S_j}(R'_j)| \leq \frac{2}{0.31} \cdot  |\IPP_{S_j}(R'_j)| $.
\end{claim}
\begin{proof}
  By \lemmaref{lemma:nB}, the restriction of the path lengths by $\pmax$ only reduces
  the fractional throughput by a factor less than $0.31$.  By \theoremref{thm:IPP}, the \IPP\
  algorithm is $(2,k)$-competitive, and hence its throughput is half the optimal
  fractional throughput with bounded path lengths.
\end{proof}

Let $|\initroute(R'_j)|$ denote the number of requests
that are accepted by $\initroute$ if invoked in isolation with
the input sequence $R'_j$.
\begin{claim}\label{claim:init}
$|f^*_{S_j}(R'_j)| \leq 2 \cdot |\initroute(R'_j)| $.
\end{claim}
\begin{proof}
  A far request must exit the tile in which it begins.  The edge capacities in the
  sketch graph are unit. Hence, the amount of flow in $f^*_{S_j}(R'_j)$ that
  originates in each tile is at most $2$.  On the other hand, if a positive amount of
  flow originates in a tile $s$, then at least one request starts in the SW-quadrant
  of $s$. Hence $\initroute(R'_j)$ accepts at least one request that begins in $s$.
\end{proof}

A na\"{\i}ve analysis of the requests accepted by the
conjunction of the \IPP\ and \initroute\ algorithms implies
that the accepted requests are in the intersection, which
might be empty. However, in our algorithm the subsequence
of accepted requests is determined by both algorithms, and
this set of accepted requests determines the state of both
algorithms. Hence, by applying the combining analysis of
Kleinberg and Tardos~\cite{KT}, the combined competitive
ratio is shown to be the sum of the isolated competitive
ratios.

\begin{claim}\label{claim:combine}
  $f^*_{S_j}(R'_j) \leq \left( 2 + \frac{2}{0.31}\right) \cdot |\alg(R'_j)| $.
\end{claim}
\begin{proof}
  To simplify notation, let $A_1$ denote the \IPP\ algorithm in isolation, $A_2$
  denote the \initroute\ algorithm in isolation, and $\alg$ denote the combined
  algorithm. Let $I$ denote the input sequence $R'_j$.

  Consider the execution of $\alg$ with the input sequence $I$.  Let $X$ denote the
  subsequence of requests accepted by the combined algorithm $A$ on input $I$ (i.e.,
  $X=A(I)$).  Let $X_j$ denote the subsequence of requests accepted by algorithm
  $A_j$ during this same execution. Note that $X=X_1\cap X_2$.  Let us rewrite
  $f^*_{S_j}(I)$ as a function $f^*: I \rightarrow [0,1]$ in which $f^*(i)$ equals
  the amount of flow assigned to request $i\in I$ by $f^*_{S_j}(I)$. We abuse
  notation and view $X$ also as its characteristic function (i.e., $X(i)=1$ if and
  only if $i\in X$).

  Consider the packet routing version in which requests have demands in $[0,1]$. A
  demand of $1$ corresponds to the situation till now. A demand of zero means that
  the request does not appear (the input skips over this request). A fractional
  demand means that (i)~at most this fraction can be routed by the fractional
  routing, and (ii)~the request only occupies this fraction of the capacity of an
  edge. Fractional demand functions can be added and multiplied. If both
  functions are integral, then addition corresponds to the union, and
  multiplication corresponds to the intersection.

  Consider the demand function $Z_{\ell}\DEF X+f^*\cdot (1-A_{\ell})$, for
  $\ell\in\{1,2\}$. First, note that it attains values in $[0,1]$ because (i)~an
  accepted request $i\in X$ satisfies $A_{\ell}(i)=1$, and (ii)~a rejected request
  $i\notin X$ satisfies $X(i)=0$.  Note also that $f^*\cdot (X+1-A_{\ell})$ is a
  feasible fractional flow with respect to the demand $Z_{\ell}$, hence
   \begin{align}
     \label{eq:1}
     |f^*_{S_j}(Z_{\ell})| &\geq \|f^*\cdot (X+1-A_{\ell})\|_1.
   \end{align}

   We claim that the isolated algorithm $A_{\ell}$ on input $Z_{\ell}$ accepts
   exactly $X$.  To prove this consider the sequence of states of
   $A_{\ell}(Z_{\ell})$ and $A(X)$. Let $S^i_{\ell}$ denote the prefix of accepted
   requests by $A_{\ell}(Z_{\ell})$ till but not including the arrival of request
   $i$.  Similarly, let $X^i$ denote the prefix of accepted requests by $A(X)$ till
   but not including the arrival of request $i$. If $S^i_{\ell}=X^i$, for every $i$, then
   $A_{\ell}(Z_{\ell})=X$. Now we prove that $S^i_{\ell}=X^i$ by induction on $i$. Before
   the first request, both sets are empty. The induction step for $i\in X$ is easy
   because $A_{\ell}$ accepts $i$ when the state is $X^i$. On the other hand, if
   $i\notin X$ and $Z_{\ell}(i)>0$, then $A_{\ell}(i)=0$.

Let $CR_{\ell}$ denote the competitive ratio of $A_{\ell}$. Then
\begin{align}
  |f^*_{S_j} (Z_{\ell})| &\leq CR_{\ell} \cdot A_{\ell}(Z_{\ell}) = CR_{\ell} \cdot |X|
\label{eq:2}
\end{align}

By Equations~\ref{eq:1} and~\ref{eq:2}
\begin{align}
  \label{eq:3}
   \|f^*\cdot (X+1-A_{\ell})\|_1 &\leq CR_{\ell} \cdot |X|.
\end{align}
Observe that $1\leq (X+1-A_1) + (X+1-A_2)$. Indeed, if $i\in X$, then
$A_1(i)=A_2(i)=1$, so both sides equal $1$. If $i\notin X$, then $A_1(i)$ or $A_i(2)$
equals zero (perhaps both), and hence the right hand side is at least $1$.
Hence,
\begin{align*}
   \|f^*\|_1 &\leq  \|f^*\cdot (X+1-A_1)\|_1  +  \|f^*\cdot (X+1-A_2)\|_1 \\
&\leq CR_1 \cdot |X| + CR_2 \cdot |X|,
\end{align*}
and the claim follows.
\end{proof}

\subsection{Near Requests}
In this section we analyze the competitive ratio of the \routenear\ algorithm with
respect to near requests.  Recall that: (1)~A request is a near request if
the distance from the source to the destination is at most $\vl$. Note that
$\vl=\Theta(\frac{\log n}{B})$ and $B<\log (1+3\pmax)=O(\log n)$.  (2)~The
incoming requests are filtered so that at most $B'+c'$ requests originate in every
space-time vertex.

The following theorem states that \routenear\ succeeds in routing at least a
logarithmic fraction of the filtered near requests. This theorem implies that the
throughput is at least a logarithmic fraction of the optimal fractional routing of
the filtered near requests.
\begin{theorem} \label{thm:near}
$|\algn| \geq \Omega(\frac {1}{\log n}) \cdot |\near|$.
\end{theorem}
\begin{proof}\sloppy
  It suffices to prove that $|\algn| \ge \Omega \left(\frac {1}{\log n}\right) \cdot |
  \near\setminus \algn|$.  Consider the following bipartite conflict graph.  Nodes on
  side $L$ are the requests of $\algn$, and nodes on side $R$ are the requests of
  $\near\setminus \algn$.  There is an edge $(r_i,r_j)\in L\times R$ if $r_j$ is
  rejected by the \routenear\ Algorithm and the vertical route of $r_i$ traverses the
  source vertex $(a_j,t_j)$ of $r_j$.  A request $r_i\in L$ conflicts with at most
  $B'+c'$ requests in each vertex. Hence, the degree of $r_i$ in the conflict graph is
  at most $(B'+c')\cdot \vl$.  On the other hand, the degree of $r_j\in R$ equals $c'$
  (where $c'$ is the capacity of the track reserved for the near requests).
  \medskip\noindent By counting edges on each side we conclude that
\begin{align*}
\sum_{r_i\in L} \deg(r_i) = \sum_{r_j\in R} \deg(r_j).
\end{align*}
Hence,
\begin{align*}
(B'+c')\cdot \vl \cdot |L| \geq c'\cdot |R|.
\end{align*}
We conclude that
\begin{align*}
|L| \geq \frac{c'}{(B'+c')\vl} \cdot |R|.
\end{align*}
As $\frac{(B'+c')\vl}{c'}=\Theta(\frac{\log n}{c'}+ \frac{\log n}{B'})=O(\log n)$,
and the theorem follows.
\end{proof}

\subsection{Putting Things Together}\label{sec:ptt}\label{sec:together}
In this section we prove~\lemmaref{lemma:frac cr}.  We partition the input sequence
$R'$ into $\near$ and $R'_j$, for $j\in \{1,2,3,4\}$  (recall that
$R'_j=R'\cap \far_j$). By subadditivity,
\begin{align}
  \label{eq:opt sub}
  |f^*_{\Gst} (R')|&\leq   |f^*_{\Gst} (\near)| + \sum_{j=1}^4 |f^*_{\Gst} (R'_j)|,\\
  |\alg(R')| &= |\alg(\near)| + \sum_{j=1}^4 |\alg(R'_j)|.
\end{align}
In order to bound the ratio $ |f^*_{\Gst} (R')|/ |\alg(R')|$, it suffices to
separately bound the ratios of the terms. Indeed, by~\theoremref{thm:near}
 $$|f^*_{\Gst} (\near)|\leq O(\log n) \cdot |\alg(\near)|.$$
By~\propref{prop:scale} and~\claimref{claim:combine},
\begin{align*}
  |f^*_{\Gst} (R'_j)| &\leq O(\log n) \cdot |\alg(R'_j)|.
\end{align*}
Finally, by~\propref{prop:R'}, $|f^*_{\Gst} (R)| \leq  9 \cdot |f^*_{\Gst} (R')|$.
Since $\alg(R)=\alg(R')$, the lemma follows.

\section{Extension to $d$-Dimensional Grids}\label{sec:d dim}

The following theorem is proved by extending~\Algref{alg:outline} for a
line network to $d$-dimensional grid.
\begin{theorem}\label{thm:d dim}
  For $B,c \in [2^{d+1}+1,\log n]$, there is a deterministic $2^{O(d)}\cdot \log^d
  n$-competitive online algorithm for the throughput maximization problem.
\end{theorem}
\begin{proof}[(sketch)]
  As in the one-dimensional case, perform a space-time transformation on the
  $d$-dimensional $n$-node grid $G$ to obtain the $(d+1)$-dimensional space-time grid
  $G^{st}$. Partition $G^{st}$ to $\ell_1\times\ldots \times\ell_{d+1}$ subgrids (or
  subcubes). The side length of a subgrid equals $\vl$ for directions that correspond
  to forward steps and $\Hl$ in the direction that corresponds to store steps.  There
  are two offsets per dimension, resulting with $2^{d+1}$ tilings. The number of
  tracks equals the number of offsets plus one (the extra one is for the near requests), hence we
  require that $B,c \geq 2^{d+1}+1$.  Similarly to the $1$-dimensional case, a
  request is classified as a near request if the distance from the source to the
  destination is at most $d\cdot \vl$.  Detailed routing within a tile is successful
  by the following observation. Every time a packet cannot turn to the direction that
  is dictated by its sketch path, there is a packet that did turn to its desired direction.
  Since the number of path emanating from each tile is bounded by the quadrant-side
  capacity, we conclude that every packet will eventually turn, if needed, within its
  quadrant, thus respecting its sketch path.

  Since the link capacity to track capacity ratio is $O(2^d)$, this scaling of
  capacities incurs an  $O(2^d)$ factor to the competitive ratio.  The
  sketch graph is obtained in the same way, with the exception that edge capacities
  are set to $\frac 1{d+1}$ (so that that the number of paths that \IPP\ routes out
  of a $(d+1)$-dimensional tile is at most $O(\log n)$). The ratio of edge capacities
  in $\Gst$ between adjacent faces of tiles and the capacity of the edge in the
  sketch graph is $O(\log n)^d$. This incurs an additional factor of $O(\log n)^d$
  for routing far requests due to capacity scaling.
The routing of near requests succeeds in routing at least a fraction of $d\cdot \log
n$ of the near requests. We conclude that the competitive ratio is determined by the
fasr requests, and hence the theorem follows.
\end{proof}

\section{Conclusion}
\label{sec:conc}
In this paper we presented an online deterministic packet routing algorithm.  For the
one dimensional grid (with constant-size buffers and constant-capacity links), this
algorithm closes the gap with the best throughput achieved by a randomized
algorithm.  This closes a problem which was open for more than a decade, but still
leaves open quite a few problems. The most urgent one is to reduce the gap between
the upper and lower bounds on the competitive ratio. Currently the best upper bound
is $O(\log n)$ for the line, and we are not aware of any no non-trivial lower bound.
We note that reducing the upper bound to $o(\log n)$ seems to require new techniques,
as the reduction to online path packing introduces a logarithmic factor in the
competitive ratio.

Another important question is to come up with reasonable distributed
algorithms. Even though, as mentioned above, the SDN model shifts many
network operation tasks to the centralized setting, it is very
interesting to find out what can be done without a central
coordinator.

\bibliographystyle{abbrv}
\bibliography{packet}
\appendix
\section{$\sqrt{\log n}$-competitiveness of initial routing}
\begin{lemma}\label{lemma:SW}
  Fix a tile $s$ and let $q$ denote its SW quadrant.
  Suppose that the sources of $m$  path requests are in $q$. Then
  $\Omega(\sqrt{m})$ path requests are served by the initial
  routing in $q$.
\end{lemma}
\begin{proof}
%
We restrict attention to rows of $q$ which contain at least $B'$
  sources and  columns of $q$ which contain at least $c'$ sources of requests
  which cannot be
  routed horizontally. Since all other packets are trivially routed by the
  algorithm, we may assume w.l.o.g.\ that there are no other packets
  with sources in $q$.

{Let $y$ denote the number of rows that contain a source vertex of
  an initial request, and let $x$ denote the number of columns that
  contain a source vertex of an initial request that is not routed
  horizontally.}
  Clearly, $m\leq xy (B'+c')$. On the other hand, detailed routing in
  $q$ serves $yB'+xc'$ requests. We now prove that $yB'+xc' =
  \Omega(\sqrt{xy (B'+c')})$.

Without loss of generality, assume that
  $c'\geq B'$. Thus it suffices to prove that
\begin{align*}
  \sqrt{\frac yx}\cdot B'+\sqrt{\frac xy}\cdot c' = \Omega(\sqrt{c'}).
\end{align*}

We proceed with case analysis.  If $y\leq x$, then $\sqrt{\frac
  yx}\cdot B'+\sqrt{\frac xy}\cdot c'\geq c'$, as
  required.
Otherwise $y>x$.  We further distinguish between two cases:
  \begin{enumerate}
  \item If $y/x\geq c'/B'$, then $\sqrt{\frac yx}\cdot B'+\sqrt{\frac xy}\cdot c'\geq
    \sqrt{\frac{c'}{B'}}\cdot B' \geq \sqrt{c'}$, as required.
  \item If $x/y > B'/c'$, then $\sqrt{\frac yx}\cdot
      B'+\sqrt{\frac xy}\cdot c'\geq
      \sqrt{\frac{B'}{c'}}\cdot c' \geq \sqrt{c'}$, as
      required.
  \end{enumerate}
\end{proof}

Note that, if a single requested is input to a SW-quadrant,
then intial routing accepts it.
\end{document}